\documentclass[11pt,a4paper]{article}
\usepackage[hmargin=2.75cm,top=2.9cm,bottom=3.3cm,nohead,footskip=40pt,marginparwidth=2.8cm,marginparsep=.2cm]{geometry}
\usepackage{amsmath,amssymb,amsthm}
\usepackage{color}
\usepackage{tikz}
\tikzstyle{normalNode}=[circle, fill=black, draw, scale=0.75]
\tikzstyle{normalEdge}=[very thick]

 \usepackage[dvips]{epsfig}
\usepackage{xspace,cite}
\usepackage{bbm}
\usepackage{algorithm,algorithmic}
\floatname{algorithm}{Algorithm}

\newtheorem{lemma}{Lemma}[section]

\newtheorem{theorem}{Theorem}[section]

\newcommand{\classNP}{{\sf NP}}
   
   \def\1{{\mathbbm 1}}

   \def\R{{\mathbb R}}
   \def\N{{\mathbb N}}

   \def\cB{{\mathcal B}}
   \def\cC{{\mathcal C}}

   \def\cI{{\mathcal I}}

   \def\cM{{\mathcal M}}

   \def\cP{{\mathcal P}}

   \def\cS{{\mathcal S}}

\title{Resource Buying Games}

\date{}

\author{%
  Tobias Harks\thanks{Maastricht University, Email: \texttt{t.harks@maastrichtuniversity.nl}.
  } \and Britta Peis\footnotemark[1]~\thanks{Technical University Berlin,
  Email: \texttt{peis@math.tu-berlin.de}.}
}

\begin{document}

\maketitle

\begin{abstract}
In \emph{resource buying games} a set of players  jointly buys a subset of a finite resource set $E$ (e.g., machines, edges, or nodes in a digraph).
The cost  of a resource $e$ depends on the number (or load) of players using $e$,
and has to be paid completely by the players before it becomes available.
Each player $i$ needs at least one set of a predefined family $\cS_i\subseteq 2^E$ to be available.
Thus, resource buying games can be seen as a variant of congestion games in which the load-dependent costs
of the resources can be shared arbitrarily among the players.
A strategy of player $i$ in resource buying games is a tuple consisting of one of $i$'s desired configurations $S_i\in \cS_i$
together with a payment vector $p_i\in \R^E_+$ indicating how much  $i$ is willing to contribute
towards the purchase of the chosen resources.
In this paper, we study the existence and computational complexity of pure Nash equilibria (PNE, for short)
of resource buying games.
In contrast to classical congestion games for which equilibria are guaranteed to exist,
 the existence of equilibria in resource buying games strongly depends on the underlying structure of the families $\cS_i$
and the behavior of the cost functions. We show that for marginally non-increasing cost functions,
matroids are exactly the right structure to consider,
and that resource buying games with marginally non-decreasing cost functions
always admit a PNE.

\end{abstract}

\section{Introduction}\label{sec:intro}

We introduce and study \emph{resource buying games} as a means to model
selfish behavior of players jointly designing a resource infrastructure.
In a resource buying game, we are given a finite set $N$ of players and a finite set of resources $E$.
We do not specify the type of the resources, they can be just anything (e.g.,
edges or nodes in a digraph, processors, trucks, etc.).
In our model, the players jointly buy a subset of the resources.
Each player $i\in N$ has a predefined family of subsets  (called  \emph{configurations})
$\cS_i\subseteq 2^E$
from which player $i$ needs at least one set $S_i\in \cS_i$ to be available.
For example, the families $\cS_i$ could be the
collection of all paths linking two player-specific terminal-nodes $s_i,t_i$ in a digraph $G=(V,E)$, or 
$\cS_i$ could stand for the set of machines on which $i$  can process her job on.
The cost $c_e$ of a resource $e\in E$ depends on the number of players using $e$,
and needs to be paid completely by the players before it becomes available.
As usual, we assume that the cost functions $c_e$ are non-decreasing and normalized in the sense that
$c_e$ never decreases with increasing load, and that $c_e$ is zero if none of the players is using $e$.
In a weighted variant of resource buying games, each player has a specific weight (demand) $d_i$,
and the cost $c_e$ depends on the sum of demands of players using $e$.
In  resource buying games, a strategy of player $i$ can be regarded as a tuple $(S_i, p_i)$ consisting of one of $i$'s  desired sets $S_i\in \cS_i$,
together with a payment vector $p_i\in \R_+^E$ indicating how much $i$ is willing to contribute towards the purchase of the resources.
The goal of each player is to pay as little as possible  by ensuring that the bought resources  contain at least one of her desired configurations.
A \emph{pure strategy Nash equilibrium} (PNE, for short)
is a strategy profile $\{(S_i, p_i)\}_{i\in N}$ such that none of the players has an incentive to switch her strategy given that
the remaining players stick to the chosen strategy.
A formal definition of the model will be given in Section \ref{subsec:prel}.

\subsubsection*{Previous Work.}
As the first seminal paper in the area of resource buying games, 
Anshelevich et al.~\cite{AnshelevichDTW08} introduced
\emph{connection games} to model selfish behavior of players
jointly designing a network infrastructure.
In their model, one is given an undirected graph $G=(V,E)$
with non-negative (fixed) edge costs $c_e, e\in E,$ and the players jointly design the network infrastructure by 
buying a subgraph $H\subseteq G$. An edge $e$ of $E$ is bought
if the payments of the players for this edge cover the cost $c_e$, and,
a subgraph $H$ is bought if every $e\in H$ is bought. Each player $i\in N$ has a specified  source node $s_i\in V$ and terminal node
$t_i\in V$ that she wants to be connected in the
bought subgraph.
A strategy of a player is a payment vector indicating how much she contributes towards the purchase of each edge in $E$.  
Anshelevich et al.\
show that these games have a PNE if all
players connect to a common source. They also show that general connection games might fail to have a PNE
(see also Section~\ref{sec:insights} below). 
Several follow-up papers (cf.\cite{Anshelevich09,AnshelevichC09,AnshelevichK11,EpsteinFY09,CardinalH10,Hoefer09,Hoefer10}) study the 
existence and efficiency of pure Nash and strong equilibria in connection games and extensions of them. In contrast
to these works, our model is more general as we assume load-dependent
congestion costs and weighted players. Load-dependent cost functions play an important role in many real-world applications
as, in contrast to  fixed cost functions, they take into account the
intrinsic coupling between the quality or cost of the 
resources and the resulting demand for it. A prominent 
example of this coupling arises in the design of telecommunication networks,
where the installation cost 
depends on the installed bandwidth  which in turn
should match the demand for it.

Hoefer~\cite{Hoefer11} studied resource buying games for load-dependent non-increasing marginal cost functions generalizing fixed costs. He considers unweighted congestion games modeling cover and facility location problems.
Among other results regarding approximate PNEs and the price of anarchy/stability, he gives a polynomial time algorithm computing
a PNE for the special case, where every
player wants to cover a single element.

\subsubsection*{First Insights.}\label{sec:insights}
Before we describe our results and main ideas in detail, we give
two examples motivating our research agenda.
\begin{center}
\begin{tikzpicture}[inner sep=1mm,scale=0.8]
\tikzstyle{place}=[circle,thick,draw=black,fill=white,scale=0.75]
  \node[place] (r1) at (1, 0) {1};
  \path (r1) ++(1.3, 0) node[place] (v1) {2}
          edge[normalEdge] (v1)
     ++(1.3, 0) node[place] (v2) {3}
       ++(0, -2.5) node[rectangle, draw, thick] (v3) {$e$}
       edge[normalEdge] (r1)
       edge[normalEdge] (v1)
         edge[normalEdge] (v2)
     ++(-2.5, 0) node[rectangle, draw, thick] (v4) {$f$} 
       edge[normalEdge] (r1)
       edge[normalEdge] (v1)
         edge[normalEdge] (v2);
      \draw (r1) ++(1.2, -3.5) node {Fig. 1(a)};
    \node[place] (r) at (7, 0) {$ s_1$};
  \path (r) ++(-1.3, -1.3) node[place] (v1) {$ s_2$}
       edge[normalEdge] node[above, sloped] {$1$} (r)
       edge[normalEdge] node[below, sloped] {$a$} (r) 
          ++(2.6, 0) node[place] (v3) {$ t_2$}
       edge[normalEdge] node[above, sloped] {$1$} (r)
        edge[normalEdge] node[below, sloped] {$b$} (r)
     ++(-1.3, -1.3) node[place] (v4) {$ t_1$} 
      edge[normalEdge] node[below, sloped] {$1$} (v1)
      edge[normalEdge] node[above, sloped] {$d$} (v1)
     %  edge[normalEdge] node[below, sloped] {$b$} (r)
      edge[normalEdge] node[below, sloped] {$1$} (v3)
      edge[normalEdge] node[above, sloped] {$c$} (v3);;
  \draw (r) ++(0, -3.5) node {Fig. 1 (b)};
       \end{tikzpicture}
\end{center}
%\caption{\label{fig:scheduling} Scheduling game (a) and connection game (b).}
%\end{figure}
Consider the scheduling game  illustrated in Fig.1(a) with 
two resources (machines) $\{e, f\}$ and three players $\{1,2,3\}$ each having 
unit-sized jobs. Any job
fits on any machine, and the processing cost
of machines $e,f$  is given by
$c_j(\ell_j(S))$, where $\ell_j(S)$ denotes the number of jobs
on machine $j\in\{e,f\}$ under schedule $S$.
In our model, each player chooses a strategy which is a tuple consisting of one of the two machines,
together with a payment vector indicating how much she is willing to pay
for each of the machines.
Now, suppose the cost functions for the two machines
are  $c_e(0)=c_f(0)=0$, $c_e(1)=c_f(1)=1$, $c_e(2)=c_f(2)=1$ and $c_e(3)=c_f(3)=M$
for some large $M>0$. One can easily verify that there is no PNE:
If two players share the cost of one machine, then a player
with positive payments deviates to the other machine.
By the choice of $M$, the case that all players share a single machine can never be a PNE.
In light of this quite basic example, we have to restrict the set of feasible cost functions.
Although the cost functions $c_e$ and $c_f$ of the machines in this scheduling game  are monotonically non-decreasing, their marginal cost function 
is neither  non-increasing, nor non-decreasing, 
where we call
 cost function
$c_e:\N\to \R_+$   \emph{marginally non-increasing [non-decreasing]} if
\begin{equation}\label{eq.discrete-concave}
 c_e(x+\delta)-c_e(x)\ge ~[\le]~c_e(y+\delta)-c_e(y)\quad\forall x\le y;~ x,y,\delta\in \N.
\end{equation}
Note that cost functions with non-increasing marginal costs model economies of scale
and include fixed costs as a special case.
Now suppose that marginal cost functions are non-increasing and consider
scheduling games on restricted machines with uniform jobs. It is not hard to establish
a simple polynomial time algorithm to compute a PNE for 
this setting:
Sort the machines with respect to the costs evaluated at load one.
Iteratively, let the player whose minimal cost among her available resources is maximal
exclusively pay for that resource, drop this player from the list and update the cost
on the bought resource with respect to a unit increment of load.

While the above algorithm might give hope for obtaining a more general existence and computability
result for PNEs for non-increasing marginal cost functions, we recall
a counter-example given by~\cite{AnshelevichDTW08}.
Consider the connection game illustrated in Fig.1(b), where there are two players
that want to establish an $s_i$-$t_i$ path for $i=1,2$. 
Any strategy profile (\emph{state}) of the game  contains two paths, one for each player, that have exactly
one edge $e$ in common. In a PNE, no player would ever pay a positive amount
for an edge that is not on her chosen path. Now, 
a player paying a positive amount for $e$ (and at least one such player exists)
would have an incentive to switch strategies as she could use the edge that is exclusively used (and paid) by the
other player for free. 
 Note that this example uses fixed costs which are marginally non-increasing.

\subsubsection*{Our Results and Outline.}\label{sec:our-results}
We study unweighted and weighted resource buying games and 
investigate the existence and computability of pure-strategy
Nash equilibria (PNEs, for short).
In light of the examples illustrated in Fig.1,
we find that equilibrium existence is strongly related
to two key properties of the game: the monotonicity of the marginal cost functions and the combinatorial structure of the allowed strategy spaces of the players.

We first consider non-increasing marginal cost functions and
investigate the combinatorial structure of the strategy spaces
of the players for which PNEs exist.
As our main result we show that \emph{matroids} are exactly the right structure to consider in this setting:
In Section~\ref{sec:matroid-algo}, we present
a polynomial-time algorithm to compute a PNE for unweighted matroid resource buying games.
This algorithm can be regarded as a far reaching, but highly non-trivial extension of the simple algorithm for scheduling games described before:
starting with the collection of matroids,
our algorithm iteratively makes use of deletion and contraction operations
to minor the matroids, until a basis together with a suitable payment vector for each of the players  is found.
The algorithm works not only for fixed costs, but also for the more general marginally non-increasing  cost functions. Matroids have a rich combinatorial structure and include, for instance,
the setting where each player wants to build a spanning tree in a graph. 
In Section~\ref{sec:weighted-matroids}, we study weighted resource buying games.  
We prove that for non-increasing marginal costs and matroid structure, every (socially) optimal configuration profile
can be obtained as a PNE. The proof relies on a complete
characterization of configuration profiles that can appear as a PNE.
We lose, however, polynomial running time as
computing an optimal configuration profile is NP-hard even for simple matroid games
with uniform players. 
In Section~\ref{sec:non-matroids}, we show that our existence result is "tight"
by proving that the matroid property is also the maximal property of the configurations of the players that leads to the existence
of a PNE: For every two-player weighted resource buying game
having non-matroid set systems, we construct
an isomorphic game that does not admit a PNE.

We finally turn in Section~\ref{sec:discrete-convex} to resource buying games having non-decreasing marginal costs.
%The assumption of non-decreasing marginal costs frequently appears 
%in transportation and telecommunication systems, where congestion costs
%increase super-linearly in the load.
We show that every such game possesses a PNE
regardless of the strategy space. We prove this result by showing that
an optimal configuration profile can be obtained as a PNE.
We further show that one can compute a PNE efficiently
whenever one can compute a best response efficiently.
Thus, PNE can be efficiently computed even in 
multi-commodity network games.

\subsubsection*{Connection to Classical Congestion Games.}
We briefly discuss connections and differences between 
resource buying games and classical congestion games. 
Recall the congestion game model: the strategy space of each player $i\in N$
consists of a family $\cS_i\subseteq 2^E$ of a finite set of resources $E$.
The cost $c_e$ of each resource $e\in E$ depends on the number of players using $e$.
In a classical congestion game, each player $i$ chooses one set $S_i\in \cS_i$ and needs to pay the \emph{average} cost of every resource in $S_i$. Rosenthal~\cite{Rosenthal:1973}
proved that congestion games always have a PNE.
This stands in sharp contrast to resource buying games for which
PNE need not exist even for unweighted singleton two-player games with non-decreasing costs,
see Fig.1(a).
For congestion games with weighted players, Ackermann et al.\cite{Ackermann09} showed that for non-decreasing marginal cost functions matroids are
the maximal combinatorial structure of strategy spaces admitting PNE.
In contrast, Theorem~\ref{t.convex-costs} shows that resource buying games with non-decreasing marginal cost functions
always have a PNE \emph{regardless} of the strategy space.
Our characterization of matroids as the maximal
combinatorial structure admitting PNE for resource buying games with non-increasing marginal costs is also different to the one of Ackermann et al.~\cite{Ackermann09} for classical weighted matroid congestion games with non-decreasing marginal costs. Ackermann et al. prove the existence
of PNE by using a potential function approach. Our existence result relies
on a complete characterization of PNE implying that there exist payments so that the optimal profile becomes a PNE. 
For unweighted matroid congestion games, Ackermann et al.~\cite{Ackermann:2008}
prove polynomial convergence of best-response by using a (non-trivial) potential function
argument. Our algorithm and its proof of correctness
are completely different relying on matroid minors and cuts.

These structural differences between the two models become even more obvious in light
of the computational complexity of computing a PNE.
In classical network congestion games with non-decreasing marginal costs it is PLS-hard
to compute a PNE~\cite{Fabrikant04,Ackermann:2008} even for unweighted players.
For network games with weighted players and non-decreasing marginal costs, Dunkel and Schulz~\cite{DunkelS08} showed that it is NP-complete to decide whether a PNE exists.
In resource buying (network) games with non-decreasing marginal costs one can compute a PNE in polynomial time even with weighted players (Theorem~\ref{t.convex-costs-poly}).

\section{Preliminaries}\label{subsec:prel}
\paragraph{The Model.}
A tuple $\mathcal{M} = (N, E, \mathcal{S}, (d_i)_{i \in N}, (c_r)_{r \in E})$ is called a
\emph{congestion model}, where $N = \{1,\dots,n\}$ is the set of players, $E = \{1,\dots,m\}$ is the set of
\emph{resources}, and  $\mathcal{S} = \times_{i \in N} \mathcal{S}_i$ is a set of states (also called \emph{configuration profiles}).
For each player $i \in N$, the set
 $\mathcal{S}_i$ is a non-empty set of subsets
$S_i \subseteq E$, called \emph{the configurations of $i$}. If $d_i=1$ for all $i\in N$ we obtain an \emph{unweighted} game,
otherwise, we have a \emph{weighted} game.
We call a configuration profile $S \in \cS$  \emph{(socially) optimal} if
its total cost $c(S)=\sum_{e\in E} c_e(S)$ is minimal among all $S\in \cS$.

Given a state  $S\in\mathcal{S}$, we define $\ell_e(S)=\sum_{i\in N:
e \in S_i} d_i$ as the total load of $e$ in $S$. Every
resource $e \in E$ has a \emph{cost function} $c_e : \cS \rightarrow \N
$ defined as $c_e(S)=c_e(\ell_e(S))$. In this paper, all cost
functions are non-negative, non-decreasing and normalized in the sense that $c_e(0)=0$.  
We now obtain a \emph{weighted resource buying game} as the (infinite) strategic game $G =
(N,\mathcal{S}\times\mathcal{P},\pi)$, where $\mathcal{P}= \times_{i \in N}
\cP_i$ with $\cP_i=\R_+^{|E|}$ is the set of feasible payments for the players. 
Intuitively, each player 
chooses a configuration $S_i\in\mathcal{S}_i$ and a payment 
vector $p_i$ for the resources.
We say that a resource $e\in E$ is \emph{bought} under
strategy profile $(S,p)$, if $\sum_{i\in N} p_i^e \geq c_e(\ell_e(S))$,
where $p_i^e$ denotes the payment of player $i$ for resource $e$.
Similarly, we say that a subset $T\subseteq E$ is bought if every $e\in T$ is bought.
The private cost function of each player $i\in N$ is defined  as
$\pi_i(S)=\sum_{e\in E} p_i^e$ if $S_i$ is bought, and $\pi_i(S)=\infty$, otherwise.
% \[ \pi_i(S,p)=\begin{cases} & \sum_{e\in E} p_i^e, \text{if $S_i$ is bought, }  \\
% & \infty, \text{ else.}
% \end{cases}
% \]
We are interested in the existence of pure Nash equilibria, i.e., strategy profiles that are resilient against unilateral deviations. 
Formally, a strategy profile $(S,p)$ is a \emph{pure Nash equilibrium}, PNE for short, if $\pi_i(S,p) \leq \pi_i((S'_i, S_{-i}), (p_i', p_{-i}))$ for all players $i \in N$ and all strategies $(S_i,p_i) \in \cS_i \times \cP_i$.
Note that  for PNE, we may assume w.l.o.g that a pure strategy $(S_i, p_i)$ of player $i$  satisfies $p_i^e\geq 0$ for all $e\in S_i$ and $p_i^e= 0$, else.

\paragraph{Matroid Games.}
We call a weighted resource buying game 
a \emph{matroid (resource buying) game} if each configuration set $\cS_i\subseteq 2^{E_i}$ with $E_i\subseteq E$  forms the base set of some matroid $\cM_i=(E_i, \cS_i)$.
As it is usual in matroid theory, we will throughout write $\cB_i$ instead of $\cS_i$, and $\cB$ instead of $\cS$,
when considering  matroid games.
Recall that 
a non-empty anti-chain\footnote{Recall that $\cB_i\subseteq 2^{E_i}$
is an \emph{anti-chain} (w.r.t.\ $(2^{E_i}, \subseteq)$) if $B,B'\in \cB_i,~ B\subseteq B'$ implies $B=B'$.} $\cB_i\subseteq 2^{E_i}$ is the base set of a matroid $\cM_i=(E_i, \cB_i)$ on resource (\emph{ground}) set $E_i$
if and only if the following \emph{basis exchange property} is satisfied:
whenever $X,Y\in \cB_i$ and $x\in X\setminus{Y}$, then  
there exists some $y\in Y\setminus{X}$ such that $X\setminus{\{x\}}\cup \{y\}\in \cB_i.$
For more about matroid theory, the reader is referred to~\cite{Schrijver:2003}.

%%%%%%%%%%%%%%%%%%%%%%%%%%%%%%%%%%%%%%%%%5
\section{An Algorithm for Unweighted Matroid Games
%with Discrete Concave Costs
}\label{sec:matroid-algo}

Let $M=(N,E, \cB, (c_e)_{e\in E})$ be a model of an unweighted matroid resource buying game. Thus, 
%for each player $i\in N$, the demand $d_i$ equals $1$,  and  
$\cB=\times_{i\in N} \cB_i$ where each
$\cB_i$ is the base set of some matroid $\cM_i=(E_i, \cB_i)$, and $E=\bigcup_{i\in N} E_i$.
In this section, we assume that the cost functions $c_e,~ e\in E$ are marginally non-increasing. 

Given a matroid $\cM_i=(E_i, \cB_i)$, we 
denote by $\cI_i=\{I\subseteq E\mid I\subseteq B \mbox{ for some } B\in \cB_i\}$
the collection of \emph{independent sets} in $\cM_i$. Furthermore, we
call
a set $C\subseteq E_i$  a \emph{cut} of matroid $\cM_i$
if $E_i\setminus{C}$ does not contain a basis of $\cM_i$.
Let $\cC_i(\cM_i)$ denote the collection of all inclusion-wise minimal cuts of $\cM_i$.
We will need the following basic insight at several places.
\begin{lemma}\cite[Chapters 39 -- 42]{Schrijver:2003}\label{lem:one-to-one-exchange}
Let $\cM$ be a weighted matroid with weight function $w:E\rightarrow \R_+$.
A basis $B$ is a minimum weight basis of $\cM$ if and only
if there exists no basis $B^*$ with $|B\setminus B^*|=1$ and
$w(B^*)<w(B)$.
\end{lemma}

In a strategy profile $(B,p)$ of our game with $B=(B_1, \ldots, B_n) \in \cB$
(and $n=|N|$) players will jointly buy a subset of
resources $\bar{B}\subseteq E$ with $\bar{B}=B_1\cup \ldots \cup B_n$.
Such a strategy profile $(B,p)$ is a PNE
if none of the players $i\in N$ would need to pay less
by switching to some other basis $B_i'\in \cB_i$,
given that all other players $j\neq i$ stick to their chosen strategy $(B_j, p_j)$.
By Lemma~\ref{lem:one-to-one-exchange}, it suffices to consider bases $\hat{B}_i\in \cB_i$
with $\hat{B}_i=B_i-g+f$ for some $g\in B_i\setminus{\hat{B}_i}$ and $f\in \hat{B}_i\setminus{B_i}$.
Note that by switching from $B_i$ to $\hat{B}_i$,
player $i$ would need to pay the additional marginal cost $c_f(l_f(B)+1)-c_f(l_f(B))$,
but would not need to pay for element $g$.
Thus, $(B,p)$ is a PNE iff for all $i\in N$ and all $\hat{B}_i\in \cB_i$
with $\hat{B}_i=B_i-g+f$ for some $g\in B_i\setminus{\hat{B}_i}$ and $f\in \hat{B}_i\setminus{B_i}$ holds
%$$\pi_i(B,p)=\sum_{e\in E} p_i^e=\sum_{e\in B_i} p_i^e\le p_i(B_i)-p_i^g +c_f(l_f(B)+1)-c_f(l_f(B)),$$
%which is equivalent to 
$p_i^g \le c_f(l_f(B)+1)-c_f(l_f(B)).$

We now give a polynomial time algorithm (see Algorithm~\ref{alg:matroid} below) computing a PNE
for unweighted matroid games with marginally non-increasing costs.
The idea of the algorithm can roughly be described as follows: 
In each iteration,  for each player $i\in N$, the algorithm maintains some independent set $B_i\in \cI_i$,
starting with $B_i=\emptyset$, as well as some payment vector $p_i\in \R^E_+$, starting with
the all-zero vector. It also maintains a current matroid $\cM_i'=(E_i', \cB_i')$ that is obtained
from the original matroid $\cM_i=(E_i, \cB_i)$ by deletion and contraction operations (see e.g., \cite{Schrijver:2003} for the
definition of deletion and contraction in matroids.)
The algorithm also keeps track of the current marginal cost 
$c_e'=c_e(\ell_e(B)+1)-c_e(\ell_e(B))$ for each element $e\in E$ and the current sequence $B=(B_1, \ldots, B_n)$.
Note that $c'_e$ denotes the amount 
that needs to be paid if some additional player $i$ selects $e$
into its set $B_i$.
In each iteration, while there exists at least one player $i$ such that $B_i$ is not already a basis,
%(this turns out to be equivalent to $E'_i\neq \emptyset$ for the ground set of the current matroid
%$\cM_i'$),
the algorithm chooses among all cuts in  
$\cC=\{C\in \cC_i(\cM_i') \mid \mbox{ for some }i\in N\}$ an inclusion-wise minimal cut $C^*$
whose bottleneck element (i.e., the element of minimal current weight in $C^*$) 
has maximal $c'$-weight (step~\ref{Step:Choose-cut}).
(We assume that some fixed total order $(E, \preceq)$
is given to break ties, so that the choices of $C^*$ and $e^*$ are unique.)
%Note that this max-min choice implies that $C^*$ is inclusion-wise minimal among all cuts in $\cC$.
It then selects the bottleneck element $e^*\in C^*$  (step~\ref{Step:select-element}),
and some player $i^*$ with $C^*\in \cC_i(\cM_i')$ (step~\ref{Step:select-player}).
In an update step, the algorithm  lets player $i^*$ pay the
 marginal cost $c_{e^*}'$ (step~\ref{Step:Update-marginal-cost}), adds $e^*$ to $B_{i^*}$ (step~\ref{Step:Update-Basis}), 
and contracts
element $e^*$ in matroid $\cM'_{i^*}$ (step~\ref{Step:Update-Base-set}). 
If $B_{i^*}$ is a basis in the original matroid $\cM_{i^*}$, the algorithm drops $i^*$ from the player set $N$ (step~\ref{Step:Drop-Player}).
Finally, the algorithm deletes the elements in $C^*\setminus{\{e^*}\}$ in all matroids
$\cM_i'$ for $i\in N$ (step~\ref{Step:Delete-cut-from-all-matroids}), and iterates until
$N=\emptyset$, i.e., until a basis has been found for all players. 

\begin{algorithm}[h!]
  \caption{\textsc{Computing PNE in Matroids}}
  \label{alg:matroid}
\begin{algorithmic}[1]
\REQUIRE $(N,E, \cM_i=(E_i, \cB_i), c)$
\ENSURE PNE $(B,p)$
\STATE Initialize $\cB_i'=\cB_i, E_i'=E_i$, $B_i=\emptyset$, 
$p_i^e=0$, $t_e=1$, and $c_e'=c_e(1)$ 
for each $i\in N$ and each $e\in E$;
\WHILE{ $N \neq \emptyset$}
%\STATE $\cC=\{C\in \cC_i(\cM_i') \mid \mbox{ for some player }i\in N \mbox{ with } E_i'\neq \emptyset\}$
\STATE \label{Step:Choose-cut} choose  $C^*\leftarrow \mbox{argmax}\{\min\{ c'_e : e\in C\}\mid C\in\cC\text{ inclusion-wise minimal}\}$\\ where
$\cC=\{C\in \cC_i(\cM_i') \mid \mbox{ for some player }i\in N\}$; 
%($C^*$ has minimal $c'$-weight)
\STATE \label{Step:select-element} choose $e^*\leftarrow \mbox{argmin}\{c'_e\mid e\in C^*\}$; 
%($e^*$ has minimal $c'$-weight in $C^*$)
\STATE \label{Step:select-player} choose $i^*$ with $C^*\in \cC_{i^*}(\cM'_{i^*})$; 
%($i^*$ is one of the players for which $C^*$ is a minimal cut)
\STATE  $p^{e^*}_{i^*}\leftarrow c'_{e^*}$;  
%(update $i^*$'s payment vector)
\STATE \label{Step:Update-marginal-cost} $c'_{e^*}\leftarrow c_{e^*}(t_{e^*}+1) - c_{e^*}(t_{e^*})$; 
\STATE \label{Step:Update-Basis}$B_{i^*}\leftarrow B_{i^*}+e^*$; 
\IF {$B_{i^*}\in \cB_{i^*}$}
\STATE \label{Step:Drop-Player} $N\leftarrow N-i^*$;
\ENDIF
\STATE \label{Step:Update-Base-set} $\cB'_{i^*} \leftarrow \cB'_{i^*}/e^* = 
\{B\subseteq E'_{i^*}\setminus{\{e^*\}}\mid B+e^*\in \cB'_{i^*}\}$; 
%(base set of matroid
%$\cM_{i^*}/e^*=(E_{i^*}\setminus{\{e^*\}}, \cB_{i^*}/e^*)$ obtained from $\cM_{i^*}$ by contracting $e^*$)
\STATE \label{Step:Update-Ground-set} $E'_{i^*}\leftarrow E'_{i^*}\setminus{\{e^*\}}$; 
%(update the current marginal cost of $e^*$)
\STATE $t_{e^*}\leftarrow t_{e^*}+1$;
\FORALL{players $i\in N$}
\STATE  \label{Step:Delete-cut-from-all-matroids}
 $\cB'_i\leftarrow \cB'_i\setminus{(C^*\setminus{\{e^*\}})}=\{B\subseteq E'_i\setminus{(C^*\setminus{\{e^*\}})} 
\mid B\in \cB'_i \}$ 
%(non-empty, since $C^*$ was chosen as an inclusion-wise minimal cut)
\STATE \label{Step:Update-groundset} 
%$E_i'\leftarrow \{e\in E_i' \mid e\in B \mbox{ for some }B\in \cB_i'\}$; 
$E'_i\leftarrow E'_i\setminus{(C^*\setminus{\{e^*\}})}$;
\ENDFOR
\ENDWHILE
\STATE $B=(B_1, \ldots, B_n)$, $p=(p_1, \ldots, p_n)$;
\STATE {\bfseries{Return} $(B,p)$}
\end{algorithmic}
\end{algorithm}

\medskip

Obviously, the algorithm terminates after at most $|N|\cdot|E|$ iterations,
since in each iteration, at least one element $e^*$ is dropped from the ground set of
one of the players.
Note that the inclusion-wise minimal cut $C^*$ whose bottleneck element
$e^*$ has maximal weight (step~\ref{Step:Choose-cut}), as well as the corresponding player $i^*$ and
the bottleneck element $e^*$, can be efficiently found, see the appendix for
a corresponding subroutine.

%%%%%%%%%%%%%%%%%%% XXXX full version only %%%%%%%%%%%%%%%%%%%%%%%%%
\iffalse
\subsection{A Subroutine to Detect $C^*, e^*$ and $i^*$}
We might assume that  the elements of $E'=E_1'\cup \ldots \cup E_n'=\{e_1, \ldots, e_m\}$
are listed by non-increasing current $c'$-weights and such that  the tie-breaking rule
induced by the total order $(E, \preceq)$ is respected. (Note that in each iteration, only the $c'$-weight of the 
chosen bottleneck element might change to some smaller weight. The $c'$-weight of the remaining elements 
keeps the same.)
The following procedure is used to detect  an inclusion-wise minimal cut $C^*$ of some player $i^*$
with the property that the bottleneck element $e^*$ of $C^*$ has maximal $c'$-weight among
all possible cuts in this game:
Initially, set  $C^*=\{e_m\}$ and $k=m$.
If $C^*$ is a cut (i.e., if $E'-C^*$ does not contain a basis) for at least some player,
set $e^*=e_k$ and
decrease $C^*$ to some inclusion-wise minimal cut 
as follows: as long as possible, choose some element $e\in C^*-e^*$  
such that $C^*-e$ remains a cut for at least one of the players $i^*\in N$, and set $C^*\leftarrow C^*-e$;
Return $(C^*, e^*, i^*)$;
Else, update $C^*\leftarrow C^*+e_{k-1}$, and iterate with $k\leftarrow k-1$.
\fi

%%%%%%%%%%%%%%%%%%%%%%%%%%%%%%%%%%%%%%%%%%%%%%%%%%%%%%%%%%%%%%%%%%%%%%%%%%%%%%

\medskip
It is not hard to see that  Algorithm \ref{alg:matroid} corresponds exactly to the procedure described in Section~\ref{sec:insights}
to solve the scheduling game (i.e., the matroid game on uniform matroids) with  non-increasing marginal cost functions.
We show 
%\ifthenelse{\boolean{full}}{}{(in the full version in the Appendix)}
that the algorithm returns a pure Nash equilibrium also for general matroids.
As a key Lemma, we show that the current weight of the chosen bottleneck element
monotonically decreases.

\medskip
\begin{theorem}
The output $(B,p)$ of the algorithm  is a PNE.

\iffalse
satisfies

$Z_i\subseteq B_i\in \cB_i$, and for each $e\in Z_i$, the set $Z-e$ does not contain a 
basis of matroid $\cM_i$.
\fi
\end{theorem}

%\ifthenelse{\boolean{full}}{
\begin{proof}
Obviously, at termination, each set $B_i$ is a basis of matroid $\cM_i$, as otherwise, player $i$ 
would not have been dropped from $N$,
in contradiction to the stopping criterium $N=\emptyset$.
Thus, we first need to convince ourselves that the algorithm terminates, i.e.,
constructs a basis $B_i$ for each matroid $\cM_i$. However,
this follows by the definition of contraction and deletion in matroids:

To see this, we denote by $N^{(k)}$ the current
player set, and by $B_i^{(k)}$ and $\cM_i^{(k)}=(E_i^{(k)}, \cB_i^{(k)})$ the current independent set 
and matroid of player $i$ at the beginning of iteration $k$.
%Note that $E_i^{(k)}=\{e\in E_i \mid e\in B \mbox{ for some }B\in \cB_i^{(k)}\}=\emptyset$ iff $\cB_i^{(k)}=\{\emptyset\}$.
Suppose that the algorithm now chooses $e^*$ in step~\ref{Step:select-element} and player $i^*$
in step~\ref{Step:select-player}. Thus, it updates $B_{i^*}^{(k+1)}\leftarrow B_{i^*}^{(k)}+e^*$ in step~\ref{Step:Update-Basis}
and considers the base set $\cB^{(k)}_{i^*}/e^*$ of the contracted matroid 
$\cM_{i^*}^{(k)}/e^*$.
Note that for each $B\in \cB^{(k)}_{i^*}/e^*$, the set $B+e^*$ is a basis in $\cB_{i^*}^{(k)}$, and, by induction, $B+B_{i^*}^{(k+1)}$ is a basis
in the original matroid $\cM_{i^*}$.
Thus, $B_{i^*}^{(k+1)}$ is a basis in $\cM_{i^*}$ (and $i^*$ is dropped from $N^{(k)}$) if and only if $\cB^{(k)}_{i^*}/e^*=\{\emptyset\}$.
%If this is the case, then also $\cB_{i^*}^{(k+1)}=(\cB^{(k)}_{i^*}/e^*)\setminus{(C^*\setminus{\{e^*\}})}=\{\emptyset\}$,
%where $C^*$ is the cut chosen in step~\ref{Step:Choose-cut}.

Now consider any other player $i\neq i^*$ with $\cB_i^{(k)}\neq\{\emptyset\}$ (and thus $i\in N^{(k)}$).
Then, for the new base set $\cB_{i}^{(k+1)}=\cB^{(k)}_{i}\setminus{(C^*\setminus{\{e^*\}})}$ we still have $\cB_{i}^{(k+1)}\neq\{\emptyset\}$,
since otherwise $C^*\setminus{\{e^*\}}$ is a cut in matroid $\cM_i^{(k)}$, in contradiction to the
choice of $C^*$.
Thus, since the algorithm only terminates when $N^{(k)}=\emptyset$ for the current iteration $k$, it terminates
with a basis $B_i$ for each player $i$.

Note that throughout the algorithm 
it is guaranteed that the current payment vectors $p=(p_1, \ldots, p_n)$ satisfy
$\sum_{i\in N}p_i^e= c_e(\ell_e(B))$ for each 
$e\in E$ and the current independent sets $B=(B_1, \ldots, B_n)$.
This follows, since the payments are only modified in step~\ref{Step:Update-marginal-cost},
where
the marginal payment $p_{i^*}^{e^*}=c_{e^*}(\ell_{e^*}(B)+1)-c_{e^*}(\ell_{e^*}(B))$
is  assigned just before $e^*$ was selected into the set $B_{i^*}$.
Since we assumed the $c_e$'s to be non-decreasing, this also guarantees 
that each component $p_i^e$ is non-negative, and positive only if $e\in B_i$.

It remains to show that the final output $(B,p)$ is a PNE.
Suppose, for the sake of contradiction, that this were not true, i.e., that there exists some
$i\in N$ and some basis $\hat{B}_i\in \cB_i$ with $\hat{B}_i=B_i-g+f$
for some $g\in B_i\setminus{\hat{B}_i}$ and $f\in \hat{B}_i\setminus{B_i}$ such that 
$p_i^g > c_f(l_f(B+1))-c_f(l_f(B)).$
Let $k$ be the iteration in which the algorithm selects the element $g$ 
to be paid by player $i$, i.e., the algorithm updates $B_i^{(k+1)}\leftarrow B_i^{(k)}+g$.
Let $C^*=C(k)$ be the cut for matroid $\cM^{(k)}_i=(E_i^{(k)}, \cB_i^{(k)})$ chosen in this iteration.
Thus,  the set $E_i^{(k)}\setminus{C^*}$ contains no basis in $\cB_i^{(k)}$, i.e., no set $B\subseteq E_i^{(k)}\setminus{C^*}$
with $B+B_i^{(k)}\in \cB_i$. Note that the final set $B_i$ contains no element from $C^*$ other than $g$, as all elements in $C^*\setminus{\{g\}}$
are deleted from matroid $\cM_i^{(k)}/g$.
We distinguish the two cases where $f\in C^*$, and where $f\not\in C^*$.

In the first case, if $f\in C^*$, then,
since the algorithm chooses $g$ of minimal current marginal weight,
we know that $p_i^g=c_g(l_g(B^{(k)}+1)-c_g(l_g(B^{(k)})) \le c_f(l_f(B^{(k)}+1)-c_f(l_f(B^{(k)}))$. Thus, the marginal cost of $f$ must decrease
at some later point in time, i.e., $c_f(l_f(B+1))-c_f(l_f(B))<c_f(l_f(B^{(k)}+1)-c_f(l_f(B^{(k)}))$.
But this cannot happen, since $f$ is deleted from all matroids for which the algorithm has not found a basis up to iteration $k$.

However, also the latter case cannot be true: 
Suppose $f\not\in C^*$. If $f\in E_i^{(k)}$, then $\hat{B}_i\setminus{B_i^{(k)}}\subseteq E_i^{(k)}\setminus{C^*}$,
but $\hat{B}_i=\hat{B}_i\setminus{B_i^{(k)}}+B_i^{(k)}\in \cB_i$, in contradiction to $C^*$ being a cut in $\cM_i^{(k)}$.
Thus, $f$ must have been dropped from $E_i$ in some iteration $l$ prior  to $k$ by either some deletion or contraction operation.
We show that this is impossible (which finishes the proof):
A contraction operation of type $\cM_i^{(l)}\to \cM_i^{(l)}/e_l$ drops only the contracted element $e_l$ from player $i$'s ground set $E_i^{(l)}$,
after $e_l$ has been added to the current set $B_i^{(l)}\subseteq B_i$. 
Thus, since $f\not\in B_i$, $f$ must have been dropped
by the deletion operation in iteration $l$. Let $C(l)$ be the chosen cut in iteration $l$, and $e_l$ the bottleneck element.
Thus, $f\in C(l)-e_l$. 
Now, consider again the cut $C^*=C(k)$ of player $i$ which was chosen in iteration $k$.
Recall that the bottleneck element of $C(k)$ in iteration $k$ was $g$.
Note that there exists some cut $C'\supseteq C(k)$ such that $C'$
is a cut of player $i$ in iteration $l$ 
and $C(k)$ was obtained from $C'$
by the deletion and contraction operations in between iterations $l$ and $k$.
Why did the algorithm choose $C(l)$ instead of $C'$? The only possible answer is, that  the bottleneck element $a$ of $C'$
has current weight $c^{(l)}_a \le c^{(l)}_{e_l}\le c^{(l)}_f$. 
On the other hand, if $f$ was dropped in iteration $l$, then $c^{(l)}_f=c_f(l_f(B+1))-c_f(l_f(B))$.
Thus, by our assumption, $c^{(l)}_f<p_i^g=c^{(k)}_g$. However, since the cost function $c_g$ is
the marginally non-increasing,
it follows that $c^{(k)}_g \le c^{(l)}_g$. Summarizing, we yield
$c^{(l)}_a \le c^{(l)}_{e_l}\le c^{(l)}_f < c^{(k)}_g \le c^{(l)}_g$, and, in particular, $c^{(l)}_{e_l}< c^{(k)}_g$, 
in contradiction to Lemma~\ref{ob.bottleneck-weight-decreases} 
below (proven in the appendix).
%%%%%%%%%%%%%%%%%%%%%%%%%%%%%%%%%%%%%%%%%%%%%%%%%%%%%%%%%%%%%%%%
\end{proof}

\begin{lemma}\label{ob.bottleneck-weight-decreases}
Let $\hat{c}_k$ denote the current weight of the bottleneck element chosen in step~\ref{Step:select-element} of iteration $k$.
Then this weight monotonically decreases, i.e., $l<k$ implies $\hat{c}_l \ge \hat{c}_k$ for all $l,k \in \mathbb{N}$.
\end{lemma}

\section{Weighted Matroid Games}\label{sec:weighted-matroids}
For proving the existence of PNE in \emph{weighted} matroid games with non-increasing marginal costs
our algorithm presented before does not work anymore.
We prove, however, that there exists  a PNE
in matroid games with non-increasing marginal costs even for weighted demands.
To obtain our existence result, we now
derive a complete characterization of configuration profiles $B\in \cB$ in
weighted matroid games $(N,E, \cB, d,c)$
that can be obtained as a PNE. 
For our characterization, we need a few definitions:
For $B\in\cB$, $e\in E$ and $i\in N_e(B):=\{i\in N\mid e\in B_i\}$ let
$\mbox{ex}_i^e:=\{f\in E-e\mid B_i-e+f\in \cB_i\}\subseteq E$
denote the set of all resources $f$ such that player $i$ could exchange the resources $e$ and $f$
to obtain an alternative basis $B_i-e+f\in \cB_i$. Note that ex$_i(e)$ might be empty,
and that, if ex$_i(e)$ is empty, the element $e$ lies in every basis of player $i$ (by the matroid basis exchange property).
Let $F:=\{e\in E\mid e \mbox{ lies in each basis of }i \mbox{ for some }i\in N\}$ denote the set of elements
that are ``fixed'' in the sense that they must lie in one of the players' chosen basis.
%and let sc$(B):=\{e\in E\mid \mbox{ex}_i(e)=\emptyset \mbox{ for some } i\in N_e(B)\}$
%denote the set of resources $e$ such that there exists at least one player using $e$ who cannot exchange
%$e$ with any other resource.
Furthermore, we define for all $e\in E-F$ and all 
$i\in N_e(B)$ and all $f\in \mbox{ex}_i(e)$
the value $\Delta_i(B;e\rightarrow f):=c_f(\ell_f(B_i+f-e,B_{-i}))-c_f(\ell_f(B))$ which is the marginal amount
that needs to be paid in order to buy resource $f$ if $i$ switches from $B_i$ to $B_i-e+f$.
Finally, let $\Delta^e_i(B)$ be the minimal value among all $\Delta_i(B;e \rightarrow f)$ with $f\in \mbox{ex}_i(e)$. 
%The following characterization is inspired by a characterization
%given in \cite{FalkeHarks:EC2011} for a class of singleton congestion games.
The proof of the following characterization can be found in the appendix.

\begin{theorem} \label{lem:charac}
Consider a weighted matroid  resource buying game $(N,E,\cB,d,c)$. There
is a payment vector $p$ such that the strategy profile $(B,p)$ with $B\in\cB$ is a PNE if and only if 
\begin{equation}\label{eq:decharged}
c_e(B)\le \sum_{i\in N_e(B)} \Delta_i^e(B) \quad \mbox{ for all }e\in E\setminus F.
\end{equation}
\end{theorem}
Note that the above characterization holds
for arbitrary non-negative and non-decreasing cost
functions.
In particular, if property (\ref{eq:decharged}) were true, it follows from the constructive proof that
the payment vector $p$ can be efficiently computed.
The following Theorem \ref{t.weighted-matroids} states that matroid games with non-increasing marginal costs and weighted demands always possess a PNE.
We prove Theorem  \ref{t.weighted-matroids} (in the appendix) by showing  that any socially optimal configuration $B\in \cB$ satisfies (\ref{eq:decharged}).
\begin{theorem} \label{t.weighted-matroids}
Every weighted matroid resource buying game with 
marginally non-increasing cost functions possesses a PNE.
\end{theorem}

%\begin{proof}
%See Appendix.
%\end{proof}

%%%%%%%%%%%%%%%%%%% XXXX full version only %%%%%%%%%%%%%%%%%%%%%%%%%
Note that the above existence result does not 
imply an efficient algorithm for computing a PNE: It is straightforward to show that computing a socially optimal
configuration profile is \classNP-hard even for unit demands and singleton strategies.

%%%%%%%%%%%%%%%%%%%%%%%%%%%%%%%%%%%%%%%%%%%%%%%%%%%%%%%%%%%%%%%%%%%%%%%%%%%%%%%%%%%%%

\section{Non-Matroid Strategy Spaces}\label{sec:non-matroids}
In the previous section, we proved that for weighted matroid congestion games
with non-negative, non-decreasing, marginally non-increasing 
cost functions, 
there always exists a PNE.
In this section, we show that the matroid property of the configuration sets is also the 
maximal property needed
to guarantee the existence of a PNE for \emph{all} weighted resource buying games with 
 marginally non-increasing costs (assuming that
there is no a priori combinatorial structure how the
strategy spaces are interweaved). 
This result and its proof (in the appendix) is related to one of Ackermann et al.\ in \cite{Ackermann09} for the classical weighted matroid congestion games with average cost
sharing and marginally non-decreasing  cost functions. 
\iffalse
There are, however, structural differences between the two models as the one of~\cite{Ackermann09}
involves marginally non-decreasing  cost functions and an a priori defined cost sharing rule.
\fi
\iffalse
The following theorem 
shows that the matroid property is the maximal property needed to guarantee the existence of a PNE
for all non-negative, non-decreasing, marginally non-increasing  cost functions.
We show (proof in the appendix):
% \ifthenelse{\boolean{full}}{}{(in the Appendix)}:
\fi
\begin{theorem}\label{t.ackermann}
For every non-matroid anti-chain $\cS$ on a set of resources $E$, there exists
a weighted two-player resource buying game $G=(\tilde{E},(\cS_1\times \cS_2)\times\cP,\pi)$ 
having marginally non-increasing  cost functions, whose strategy spaces $\cS_1$ and $\cS_2$ 
are both isomorphic to $\cS$, so that $G$ does not possess a PNE. 
\end{theorem}

\section{Non-Decreasing Marginal Cost Functions}\label{sec:discrete-convex}

In this section, we consider non-decreasing marginal cost functions 
on weighted resource buying games in general, 
i.e., $\mathcal{S} = \times_{i \in N} \mathcal{S}_i$ is not necessarily
the cartesian product of matroid base sets anymore.
We prove that for every socially optimal state $S^*$ in a congestion model with non-decreasing marginal costs,
we can define \emph{marginal cost} payments $p^*$ that result in a PNE.
Formally, for a given socially optimal configuration profile $S^*\in \cS$ and a fixed order $\sigma=1,\dots,n$
of the players, we 
let $N_e(S^*):=\{i\in N\mid e\in  S_i^*\}$
denote the players using $e$ in $S^*$,
$N^j_e(S^*):=\{i\in N_e(S^*)\mid  i \le_{\sigma} j\}$ denote the players in $N_e(S^*)$ prior or equal to $j$ in $\sigma$, and
$\ell_e^{\leq_ j}(S^*)=\sum_{i\in N^j_e(S^*)} d_i$ denote the load of these players on $e$ in $S^*$.
Given these definitions, we
allocate  the cost $c_e(\ell_e(S^*))$ for each resource $e\in E$ among the
players in $N_e(S^*)$ by setting
$p_i^e=0$ if $e\not\in S_i^*$ and $p^i_e= 
c_e(\ell_e^{\leq j}(S^*))-c_e(\ell_e^{\leq j-1}(S^*))$ if player $i$ is the $j$-th player
in $N_e(S^*)$ w.r.t. $\sigma$.
Let us call this payment vector \emph{marginal cost pricing}.
\begin{theorem}\label{t.convex-costs}
Let $S^*$ be a socially optimal solution. Then, marginal cost pricing 
induces a PNE. (The proof is in the appendix.)
\end{theorem}
We now show that there is a simple polynomial time 
algorithm computing a PNE whenever we are able to 
efficiently compute a best-response. By simply inserting the
players one after the other using their current best-response
with respect to the previously inserted players,
we obtain a PNE. It follows that for (multi-commodity) network games
we can compute a PNE in polynomial time. 
\begin{theorem}\label{t.convex-costs-poly}
For multi-commodity network games
with non-decreasing marginal costs, there is a polynomial time algorithm computing a PNE.
\end{theorem}
The proof is straight-forward: Because payments
of previously inserted players do not change in later iterations and marginal cost functions are non-decreasing, the costs of alternative strategies only increase as more players are inserted. Thus, the resulting strategy profile is a PNE.

\section{Conclusions and Open Questions}
We presented a detailed study on the existence and computational
complexity of pure Nash equilibria in resource buying games. Our results imply that the price
of stability is always one for both, games with non-decreasing marginal costs and games with non-increasing marginal costs and matroid structure. Regarding the price of anarchy, even on games with matroid structure it makes a difference whether cost functions are marginally non-increasing  or marginally
 non-decreasing. For non-increasing marginal costs it is known that the price of anarchy is $n$ (the lower bound even holds for singleton games). On the other hand, for
non-decreasing marginal costs we can show that the price of anarchy
for uniform and partition matroids is exactly $n$, while it
is unbounded for graphical matroids even on instances with only two players.
%Further details will appear in the full version of the paper.
Convergence of best-response dynamics has not been addressed so far and deserves further research.

\bibliographystyle{plain}      
\bibliography{../master-bib}

\appendix{
\section{Omitted Proofs}

\subsection{A Subroutine to Detect $C^*, e^*$ and $i^*$}
We might assume that  the elements of $E'=E_1'\cup \ldots \cup E_n'=\{e_1, \ldots, e_m\}$
are listed by non-increasing current $c'$-weights and such that  the tie-breaking rule
induced by the total order $(E, \preceq)$ is respected. (Note that in each iteration, only the $c'$-weight of the 
chosen bottleneck element might change to some smaller weight. The $c'$-weight of the remaining elements 
keeps the same.)
The following procedure is used to detect  an inclusion-wise minimal cut $C^*$ of some player $i^*$
with the property that the bottleneck element $e^*$ of $C^*$ has maximal $c'$-weight among
all possible cuts in this game:
Initially, set  $C^*=\{e_m\}$ and $k=m$.
If $C^*$ is a cut (i.e., if $E'-C^*$ does not contain a basis) for at least some player,
set $e^*=e_k$ and
decrease $C^*$ to some inclusion-wise minimal cut 
as follows: as long as possible, choose some element $e\in C^*-e^*$  
such that $C^*-e$ remains a cut for at least one of the players $i^*\in N$, and set $C^*\leftarrow C^*-e$;
Return $(C^*, e^*, i^*)$;
Else, update $C^*\leftarrow C^*+e_{k-1}$, and iterate with $k\leftarrow k-1$.

\subsection{Proof of Lemma~\ref{ob.bottleneck-weight-decreases}}

\begin{proof}
For each iteration $k$ let $c_e^{(k)}$ denote the current weight of element $e$, and $\cC_i^{(k)}$ denote
the current set of inclusion-wise minimal cuts for player $i$. 
Note that for each cut $C\in \cC_i^{(k)}$ with $k>1$, there exists some $C'\in \cC_i^{(k-1)}$ such that $C\subseteq C'$ and 
$C$ is obtained from $C'$ by the contraction and deletion operations of iteration $k-1$.
For the sake of contradiction, suppose that $k$ is the first iteration such that $\hat{c}_k > \hat{c}_{k-1}$.
Let $e$ be the bottleneck element chosen in step~\ref{Step:select-element}
of iteration $k-1$.
Thus, the corresponding cut $C(k)$ that was chosen in step~\ref{Step:Choose-cut} of iteration $k$
must be obtained from some larger cut $C'$ by removing at least one element $a\in C'$ with $c_a^{(k-1)} \le \hat{c}_{k-1}=c^{(k-1)}_{e}$,
and, if equality,
with $a \prec e$.
Since the deletion operation of iteration $k-1$ removes only  elements $e'\in E$ of weight $c^{(k-1)}_{e'} \ge c^{(k-1)}_{e}$, and if equality,
those with $e' \succ e$, the element $a$ must have been dropped from $C'$ by contracting $e$, i.e., $a=e$.
Since this contraction operation touches only the matroid of the player chosen in iteration $k-1$, say $i$, it suffices to consider only the
cut sets $\cC_i^{(k)}$ and $\cC_i^{(k-1)}$ and the base sets $\cB_i^{(k)}$ and $\cB_i^{(k-1)}$ of player $i$ in iterations $k$ and $k-1$.
So far, we observed that $a\in C'\cap C(k-1)$ where $C'$ and $C(k-1)$ are both cuts in $\cC_i^{(k-1)}$, and that the  element $a$ vanishes from cut $C'$
by the contraction operation $\cM^{(k-1)}_i \rightarrow \cM^{(k-1)}_i/a$. Thus, $C'-a$ must be a (not necessarily inclusion-wise minimal) cut in $\cM^{(k-1)}_i/a$.
However, since $C'$ is an inclusion-wise minimal cut in $\cM_i^{(k-1)}$, the set
$E_i^{(k-1)}-(C'-a)$ contains some basis $\hat{B}\in \cB_i^{(k-1)}$ with $a\in \hat{B}$.
Thus, $B:=\hat{B}-a$ is a set in $E_i^{(k-1)}-(C'-a)$ with $B+a\in \cB_i^{(k-1)}$, in contradiction to $C'-a$ being a cut in $\cM^{(k-1)}_i/a$.
\end {proof}

\subsection{Proof of Theorem~\ref{lem:charac}}

\begin{proof}
We first proof the "only if" direction. Let $(B,p)$ be a PNE.
Then, by Lemma~\ref{lem:one-to-one-exchange} and the definition of a PNE, we obtain for all $e\in E\setminus F$:
\begin{align*}
c_e(B)=  \sum_{i\in N_e(B)} p_i^e \leq \sum_{i\in N_e(B)}\Delta_i^e(B).
\end{align*}
Note that the $\Delta_i^e(B)$ are well defined as we only consider elements in $E\setminus F$.
Now we prove the "if" direction. 
For all $e\in F$ we pick a player $i$ with 
%$\{e\}\in \cC_i(\cM_i/(B_i-e))$
ex$_i(e)=\emptyset$
and let her pay the entire cost, i.e., $p_i^e=c_e(B)$.
For all $e\in E\setminus F$ and $i\in N_e(B)$, we define
$$ p_i^e=\frac{\Delta_i^e(B)}{\sum_{j\in N_e(B)} \Delta_j^e(B)} \cdot c_e(B),$$
if the denominator is positive, and $p_i^e=0$, otherwise.
Using  \eqref{eq:decharged}, we obtain
\[ p_i^e\leq \Delta_i^e(B)
\text{ for all $e\in E\setminus F$}
\]
proving that $(B,p)$ is a PNE.
\end{proof}

\subsection{Proof of Theorem~\ref{t.weighted-matroids}}

\begin{proof}
We prove that any socially optimal configuration profile $B\in \cB$ 
satisfies ~\eqref{eq:decharged} and, thus, by Theorem~\ref{lem:charac} there
exists a payment vector $p$ such that $(B,p)$ is a PNE.
Assume by contradiction that $B$ does not satisfy~\eqref{eq:decharged}.
Hence, there is an $e\in E\setminus F$ with 
\begin{equation}\label{opt_weighted}
c_e(B)> \sum_{i\in N_e(B)} \Delta_i^e(B).		
\end{equation}
By relabeling indices we may write $N_e(B)=\{1,\dots k\}$ for some $1\leq k\leq n,$ and define
for every $i\in N_e(B)$
the tuple $(\hat B_i,f_i)\in \cB_i\times (E_i-e)$ as the one minimizing 
%$c_{f}(\ell_{f}(B'_i,B_{-i}))-c_{f}(\ell_f(B))$
$\Delta_i(B; e\rightarrow f)$
among all tuples $(B'_i,f) \in \cB_i\times (E_i-e)$ with $B'_i=B_i+f-e\in \cB_i$.
%\[(\hat B_i,f_i)=\arg\min_{ (B'_i,f) \in \cB_i\times (E_i-e) : B'_i=B_i+f-e\in \cB_i} \! \! \{c_{f}(\ell_{f}(B'_i,B_{-i}))-c_{f}(\ell_f(B))\}.\]
Note that  $(\hat B_i,f_i)$ is well defined as $e\in E\setminus F$.
%We set $\Delta_i(B)=c_{f_i}(\ell_{f_i}(\hat{B}_i,B_{-i}))-c_{f_i}(\ell_{f_i}(B))$ for $i\in N_e(B)$. 
We now iteratively change the current basis of every player in  $N_e(B)$ 
in the order of their indices to the alternative basis $\hat B_{i}, i=1,\dots k.$ 
This gives a sequence
of profiles $(B^0,B^1,\dots B^k)$ with $B^0=B$ and $B^{i}=(\hat B_{i}, B^{i-1}_{-i})$ for $i=1,\dots k$.
For the cost increase of the new elements $f_i, i\in N_e(B)$, we obtain the key inequality
$c_{f_i}(\ell_{f_i}(B^{i-1}))-c_{f_i}(\ell_{f_i}(B^i))\leq \Delta^e_i(B)$. This inequality holds because cost functions are marginally non-increasing, that is,
the marginal costs only decrease with higher load.
Plugging everything together, yields
\begin{align*}
c(B)-c(B^k)&=\sum_{i=1}^k (c(B^{i-1})-c(B^i))=\sum_{i=1}^k (c_e(\ell_e(B^{i-1}))+c_{f_i}(\ell_{f_i}(B^{i-1}))\\ & 
-c_e(\ell_e(B^{i}))-c_{f_i}(B^{i}))\\
& = c_e(\ell_e(B))-c_e(\ell_e(B^k)) + \sum_{i=1}^k (c_{f_i}(\ell_{f_i}(B^{i-1}))-c_{f_i}(B^{i}))\\
&\geq c_e(\ell_e(B)) - \sum_{i=1}^k \Delta^e_i(B)>0,
\end{align*}
where the first inequality uses $c_e(\ell_e(B^k))=c_e(0)=0$ (note that $e\in E\setminus F$)
and the assumption that cost functions have non-increasing marginal costs.
The second strict inequality follows from~\eqref{opt_weighted}. Altogether, we obtain 
a contradiction to the optimality of $B$.
\end{proof}

\subsection{Proof of Theorem~\ref{t.ackermann}}
Recall that $\cS\subseteq 2^E$ is  an \emph{anti-chain} (with respect to $(2^E, \subseteq)$) if for every $X\in \cS$, 
no proper superset $Y\subset X$
belongs to $\cS$.
Also note that it suffices to consider configuration sets $\cS_i$ that form an anti-chain,
as (due to the non-negative cost functions)
player $i$ would never have an incentive to switch her strategy to a superset of her chosen one. 

We call $\cS$ a \emph{non-matroid} set system if the tuple 
$(E, \{X\subseteq S : S\in \cS\})$ is not a
matroid.
The following Lemma can also be derived from the proof of Lemma 16 in \cite{Ackermann09}.
\begin{lemma}\label{l.anti-matroid}
 If $\cS\subseteq 2^E$ is a non-matroid antichain, then there exist $X,Y\in \cS$ and $\{a,b,c\}\subseteq X\cup Y$
such that each set $Z\subseteq (X\cup Y)-a$ contains both, $b$ and $c$.
\end{lemma}
%\subsection{Proof of Lemma~\ref{l.anti-matroid}}

% \begin{lemma}\label{l.anti-matroid}
%  If $\cS\subseteq 2^E$ is a non-matroid antichain, then there exist $X,Y\in \cS$ and $\{a,b,c\}\subseteq X\cup Y$
% such that each set $Z\subseteq (X\cup Y)-a$ contains both, $b$ and $c$.
% \end{lemma}

\begin{proof}
Recall the \emph{basis exchange property} for matroids: an anti-chain $\cB\subseteq 2^E$ is the family of bases of some matroid
if and only if for any $X, Y\in \cB$ and $x\in X\setminus{Y}$ there exists some $y\in Y\setminus{X}$ such that $X-x+y\in \cB$.
Thus, if the anti-chain $\cS\subseteq 2^E$ is a non-matroid, there must exist $X,Y\in \cS$ and $x\in X\setminus{Y}$
such that \emph{for all} $y\in Y\setminus{X}$ the set $X-x+y$ does \emph{not} belong to $\cS$.
We choose such $X, Y$ and $x\in X\setminus{Y}$ with $|Y\setminus{X}|$ minimal 
(among all $Y'\in \cS$ with $X-x+y'\not\in \cS$ for all $y'\in Y'\setminus{X}$).
We distinguish the two cases  $|Y\setminus{X}|=1$ and  $|Y\setminus{X}|\ge 1$:
In case $|Y\setminus{X}|=1$, set $\{a\}=Y\setminus{X}$ and choose any two distinct elements $\{b,c\}\in X\setminus{Y}$.
Note that  $|X\setminus{Y}|\ge 2$ as otherwise, if $X\setminus{Y}=\{x\}$, then $Y=X-x+a$, in contradiction to our assumption.
Now, for any set $Z\subseteq (X\cup Y)-a$, the anti-chain property implies $Z=X$, and therefore $\{b,c\}\subseteq Z$, as desired.

In the latter case $|Y\setminus{X}|\ge 1$, we choose any two distinct elements $\{b,c\}\in Y\setminus{X}$ and set $a=x$.
Consider any $Z\in \cS$ with $Z\subseteq (X\cup Y)-a$ and suppose, for the sake of contradiction, that $\{b,c\}\not\subseteq Z$.
Since $Z\setminus{X}\subseteq Y\setminus{X}$, there cannot exist some $z\in Z\setminus{X}$ with $X-a+z\in \cS$.
However, $|Z\setminus{X}| < |Y\setminus{X}|$ in contradiction to our choice of $Y$.
\end{proof}

\begin{proof}[Proof of the Theorem]
 Let $\cS_1\subseteq 2^{E_1}$ and $\cS_2\subseteq 2^{E_2}$ be the two strategy spaces
for player one and player two, respectively, both isomorphic to our given non-matroid anti-chain $\cS\subseteq 2^{E}$.
In the following, we describe the game $G$ by defining the demands and costs and describing how the
resources and strategy spaces interweave:
For each player $i=1,2$, choose $X_i, Y_i\in \cS_i$ and $\{a_i, b_i, c_i\}\subseteq E_i$
as described in Lemma~\ref{l.anti-matroid}.
In our game $G$, the two players have only three resources in common, i.e., $\{x,y,z\}= E_1\cap E_2$.
We set $x:=a_1=b_2$, $y:=a_2=b_1$ and $z:=c_1=c_2$. All other resources in $E_i\setminus{\{x,y,z\}}$ are
exclusively used by player $i$ for $i=1,2$.
We define the (load-dependent) costs $c_e(t)$ , $t\in \R_+$ for the resources $e\in \tilde{E}=E_1\cup E_2$  as follows:
all elements in $(X_1\cup X_2\cup Y_1\cup Y_2)\setminus{\{x,y,z\}}$ have a cost of zero, and all elements
in $E_1\setminus{(X_1\cup Y_1)}$ and in $E_2\setminus{(X_2\cup Y_2)}$ have some very large cost $M$.
The costs on $\{x,y,z\}$ are defined as
$c_x(t)=t, c_y(t)=5\frac{1}{2}$ and $c_z(t)=4$.
Note that each of these cost functions is non-negative, non-decreasing and marginally non-increasing .

Now, suppose that $(Z^*, p^*)$ with $Z^*=(Z_1^*, Z_2^*)\in \cS_1\times \cS_2$ and $p^*=(p_1^*, p_2^*)\in \R_+^{E_1}\times \R_+^{E_2}$
were a PNE for the game as described above with demands $d_1=5$ and $d_2=4$.
Choosing $M$ large enough ensures that $Z^*_i\subseteq X_i\cup Y_i$ for each player $i\in \{1,2\}$.
Moreover, by the choice of $X_i$ and $Y_i$ in the proof of Lemma~\ref{l.anti-matroid},
there exist $S_1, T_1\in \cS_1$ with $x\in S_1, \{y,z\}\cap S_1=\emptyset$ and $x\not\in T_1 \supseteq \{y,z\}$,
as well as $S_2, T_2\in \cS_2$ with $y\in S_2,~  \{x,z\}\cap S_2=\emptyset$ and $y\not\in T_2 \supseteq \{x,z\}$.
By Lemma~\ref{l.anti-matroid}, it follows from $Z_i^*\subseteq X_i\cup Y_i$ that
\begin{equation}
 \label{eq.non-matroid}
x\not\in Z_1^* \Longrightarrow \{y,z\}\subseteq Z_1^* \quad\mbox{ and }\quad y\not\in Z_2^* \Longrightarrow \{x,z\}\subseteq Z_2^*.
\end{equation}
We now show that neither  $x\in Z_1^*$, nor  $x\not\in Z_1^*$ can be true.
This would be the desired contradiction to our assumption that the game possesses a PNE.

For each player $i\in \{1,2\}$, and each configuration $S_i\in \cS_i$, let
$c_i^*(S_i)$ denote the price that player $i$ would have to
pay so that the resources in $S_i$ are bought, given that the other player $j\in \{1,2\}\setminus{\{i\}}$
sticks to her strategy $(Z^*_j, p^*_j)$.
Consider the case $x\not\in Z_1^*$: By (\ref{eq.non-matroid}), it follows that $\{y,z\}\subseteq Z_1^*$.
Thus, since $Z_1^*\subseteq X_1\cup Y_1$, the only resources in $Z_1^*$ of non-zero cost are $y$ and $z$, 
i.e.,
$p_1^*(Z_1^*)=p_1^*(y)+p_1^*(z)\le c_1^*(S_1)=d_1=5$.
Note  that $y\not\in Z_2^*$ is not possible, as otherwise player 1 would need to pay $c_1^*(y)=5\frac{1}{2}$ to buy resource $y$
which is more than the price of $d_1=5$ needed to buy  $S_1$.
Thus,
 $y\in Z_2^*$ must be true. It follows that $p_2^*(z)=0$,
as otherwise $p^*_2(Z_2^*)\ge p_2^*(y)+p^*_2(z) > p_2^*(y)=c_2^*(S_2)$.
Thus, since $z\in Z_1^*$, player 1 has to pay $p_2^*(z)=4$ in order to buy resource $z$.
Since $p_1^*(Z_1^*)\le c_1^*(S_1)=5$, it follows that $p_1^*(y)\le 1$, and therefore, since $y\in Z_2^*$,
$p_2^*(y)\ge 4\frac{1}{2}$. However, 
in this case player 2 could use resource $z$ for free and therefore switch to strategy $T_2$
for which she would only need to pay the price for resource $x$ which is $d_2=4$.
So, $x\not\in Z_1^*$ is not possible in a PNE.

It remains to consider case  $x\in Z_1^*$: Then $p_1^*(Z_1^*) =p_1^*(x)+p_1^*(y)+p_1^*(z)\ge p_1^*(x)=c_1^*(S_1)$
implies $p_1^*(y)=p_1^*(z)=0$. Therefore
$y\not\in Z_2^*$, since otherwise,  $p_2^*(y)=5\frac{1}{2}$,
so that player 1 could use resource $y$ for free and therefore switch to strategy $T_1\in \cS_1$ of cost $4=c_z$.
However, if $y\not\in Z^*_2$, then $\{x,z\}\subseteq Z_2^*$ by equation (\ref{eq.non-matroid}).
Hence, $p_2^*(z)=4$ (since $p_1^*(z)=0$).
It follows that $p_2^*(x) \le 1\frac{1}{2}$, since otherwise, player 2 would switch to strategy $S_2\in \cS_2$
and pay only the price of $5\frac{1}{2}$ for resource $y$.
Thus, $p_1^*(x) \ge 7\frac{1}{2}$ which is strictly greater than the price of $5\frac{1}{2}$ which player 1 would need to pay
if she switches to strategy $T_1$.
Hence, also $x\in Z_1^*$ is not possible in a PNE, which finishes the proof.

\end{proof}

\subsection{Proof of Theorem~\ref{t.convex-costs}}

\begin{proof}
Let $p=(p_1, \ldots, p_n)$ be the payment vector obtained by marginal cost pricing.
Suppose there is a player that unilaterally improves
by deviating to some $(S'_i,p'_i)$. 
Thus,  $S'=(S_1^*, \ldots, S_{i-1}^*, S_i', S_{i+1}^*, \ldots, S_n)$,
and $p_i'$ differs from $p_i$ only on elements in $S_i'\Delta S_i^*=S_i'\setminus{S_i^*}\cup S_i^*\setminus{S_i'}$, while $p_j'=p_j$ for all other players $i\neq j\in N$.
For the payoff difference, we therefore calculate that 
\[ \pi_i(S',p')-\pi_i(S^*,p)=\sum_{r\in S'_i\setminus S^*_i}  \Big(c_r(\ell_r(S^*)+d_i)-c_r(\ell_r(S^*))\Big)-\sum_{r\in S^*_i\setminus S'_i} p_i^r <0.\]
Because costs are marginally non-decreasing, we obtain
$  p_i^r\leq c_r(\ell_r(S^*))-c_r(\ell_r(S^*)-d_i)  \text{ for all }r\in S^*_i.$
 Using this inequality we obtain
 \begin{align*}c(S')-c(S^*)&=\sum_{r\in S'_i\setminus S^*_i}  \Big(c_r(\ell_r(S^*)+d_i)-c_r(\ell_r(S^*))\Big)-\sum_{r\in S^*_i\setminus S'_i}
 \Big(c_r(\ell_r(S^*))-c_r(\ell_r(S^*)-d_i) \Big)\\
 & \leq \sum_{r\in S'_i\setminus S^*_i}  \Big(c_r(\ell_r(S^*)+d_i)-c_r(\ell_r(S^*))\Big)-\sum_{r\in S^*_i\setminus S'_i}
p_i^r\\
& <0,
 \end{align*}
a contradiction to $S^*$ being optimal.
\end{proof}

\end{document}